\documentclass{article}
\usepackage{lineno,hyperref}
\modulolinenumbers[5]

\usepackage{graphicx,fancyhdr,url,amsfonts}   
\usepackage{graphicx}
\usepackage{amsthm}
\usepackage{amssymb}
\usepackage{amsmath}
\usepackage{graphicx}
\usepackage{lmodern}
\usepackage{amsmath}
\usepackage{amsfonts}
\usepackage{lmodern}
\usepackage[T1]{fontenc}
\usepackage{amssymb}
\usepackage{dsfont}
\usepackage{enumerate}
\usepackage{xcolor}
\usepackage{color}
\usepackage[normalem]{ulem} 
\usepackage{amsthm}
\usepackage{float}
\usepackage[linesnumbered,ruled,vlined]{algorithm2e}

\usepackage{amsmath} 
\usepackage{amssymb} 
\usepackage{pifont}
\usepackage{caption}
\usepackage{subcaption}
\usepackage{hyperref}
\usepackage{tkz-euclide}

\newtheorem{theorem}{Theorem}
\newtheorem{defi}{Definition}
\newtheorem{coro}{Corollary}
\newtheorem{lema}{Lemma}
\newtheorem{ejem}{Example}

\newtheorem{remark}{Remark}
\newtheorem{claim}{Claim}

\newcommand\Omit[1]{}
\newcommand\myred[1]{#1}
\newcommand\myblue[1]{#1}
\newcommand{\MA}{\mathcal{A}}
\newcommand{\MR}{\mathcal{R}}
\newcommand{\partes}[1]{{\mathcal P}({#1})}

\newenvironment{myproofof}[1]{\noindent \bf Proof of {#1}:
\setlength{\parskip}{4pt} \rm}{\ \null
\null \ \hfill \rule[-1mm]{1.2mm}{2mm}\rule[-2mm]{0mm}{4mm}\bigskip }

\newenvironment{myproof}{\noindent \bf Proof:
\setlength{\parskip}{2pt} \rm}{\ 
  \hfill \rule[-1mm]{1.2mm}{2mm}\rule[-2mm]{0mm}{4mm} }

\newcommand{\nc}{\newcommand}

\nc{\Set}[1]{\{ {#1} \}}
\nc{\ramon}[1]{{\color{blue}[{#1}]}}

\begin{document}


\title{
\bf Generalized binary utility functions and fair allocations\footnote{This work is an updated version of {\em Beyond identical utilities: buyer utility functions and fair allocations} appeared in arXiv in September 2021. This work was also published in {\em Math. Soc. Sci.} in 2023.}}

\author{Franklin Camacho$^1$, Rigoberto Fonseca-Delgado$^1$, Ramón {Pino Pérez}$^2$\\ and   Guido Tapia$^1$\\[2mm]
\em  $^1$ School of Mathematical and Computational Sciences, Yachay Tech University,\\ \em Urcuquí, Ecuador\\[1mm]
$^2$ \em Centre de Recherche en Informatique de Lens (CRIL), CNRS, Université d'Artois,\\ \em Lens, France\\[2mm]
Emails: \{fcamacho,rfonseca,guido.tapia\}@yachaytech.edu.ec, pinoperez@cril.fr
}

\maketitle

\begin{abstract}
The problem of finding envy-free  allocations of indivisible goods cannot always be solved; therefore, it is common to study some relaxations such as envy-free up to one good (EF1) 
 and envy-free up to any positively valued good (EFX). Another property of interest for {the} efficiency of an allocation is the Pareto Optimality (PO). Under additive utility functions \myred{for goods}, it is possible to find EF1 and PO {allocations} using the Nash social welfare. However,  finding an allocation that maximizes the Nash social welfare is a computationally costly problem.
Maximizing the utilitarian social welfare subject to EF1 constraints is an NP-complete problem for the case where three or more agents participate.
In this work, we propose {a restricted case of additive utility functions called \myred{generalized binary} utility functions.
The proposed utilities \myred{are a generalization  of binary and} identical  \myred{utilities simultaneously.}
In this scenario, we present} a polynomial-time algorithm {that} maximizes the utilitarian social welfare and, at the same time, produces an EF1 and PO allocation \myred{for goods as well as for chores}.  
Moreover, a slight modification of our algorithm gives a better allocation: one
which is EFX. 
\end{abstract}

\vspace{1mm}



\maketitle
\noindent

\section{Introduction}
\label{sec:intro}
 The resource allocation problem has been widely studied in mathematics and economics for almost a century, 
 \cite{Ste48,Ste49,Nash-50,Aziz-Caragiannis2018,caragiannis2019,Camacho2020}.
The main elements in the problem are agents and resources. 
The goal  is to distribute (or allocate) the resources among the agents in a ``good manner.'' 
The agents can represent individuals, objects, government institutions, among others, depending on the application. The set of resources or goods to be distributed can be divisible or indivisible.  In general, these sets are considered finite. In this work, we study the problem of allocating indivisible resources among a group of agents, with the aim of  satisfying both the group and each individual in the best possible way.


This topic  has many applications, for instance in  solving divorce disputes, dividing an inheritance, sharing apartment rents, or even assigning household chores. In the last decade, there  has been a considerable interest in the computational aspects of this problem.  In particular, in Artificial Intelligence and more specifically in MultiAgent Systems, these problems are studied  and renamed  MultiAgent Resource Allocations (MARA) problems \cite{Aziz2016,CEM17}.

Finding a correct distribution of resources consists of distributing all the resources among the agents fairly and efficiently.
To establish efficiency and some criterion of fairness, it is necessary to consider the preferences that each agent has over resources.
In general, these preferences over resources are established through additive utility functions. 

Traditionally, fairness  is established through properties such as envy-freeness {or proportionality}.
However, there are situations where it is impossible to find allocations that meet any of these properties.  
Thus, other  weaker versions of fairness, such as  envy-free up to one  good \cite{Budish2011} or proportionality up to one good \cite{conitzer2017fair} are considered. 
Although there are results that, under certain conditions, guarantee the existence of allocations with some fairness property (see \cite{caragiannis2019}), finding them is a computationally complicated problem \cite{de-keijzer-2009}. { 
Just considering fairness may not be enough, because it could imply loss in group satisfaction.}

Efficiency, also known as Pareto efficiency or Pareto optimality, is related to the group satisfaction by an allocation. 
One way to find efficient allocations is through social welfare functions.
Caragianis et al, \cite{caragiannis2019},  showed that, under additive utility functions \myred{for goods}, it is always possible to find an allocation that is Pareto optimal and envy-free up to one good.
Actually, they prove that some of  the allocations that maximize the Nash {social} welfare are Pareto optimal and envy-free up to one good.
Unfortunately, finding   allocations which maximize the Nash welfare is also an NP-hard problem  \cite{Ramezani_Endriss_2010} (see also \cite{Aziz2016}).

Searching allocations which maximize the utilitarian social welfare is in ge\-neral a more tractable problem from the computational point of view and, because of that, commonly used. A well-known result is that, under additive utility functions, allocations that maximize this social welfare are Pareto optimal (see Theorem~\ref{lema-SWU-impli-OP}),
{ although the converse is not true. Moreover, allocations that maximize this social welfare do not always satisfy fairness properties. In Example~\ref{exp:new-and-very-simple},
 an allocation that is Pareto optimal and does not maximize utilitarian social welfare is proposed; besides, we find, in this example, that no allocation that maximizes utilitarian social welfare is envy free up to one good.}

When we consider additive utility functions,  
it is possible to  define a procedure that allows to find all the allocations which maximize utilitarian social welfare, see \cite{Camacho2020}. Moreover, finding these allocations is a computationally  tractable problem. 
Actually, in this work we propose a very simple algorithm in polynomial time, for scenarios of additive utilities, that finds an allocation which maximizes the utilitarian social welfare \myred{for goods as well as for chores}.

Unlike the allocations that maximize the Nash social welfare which are EF$1$, the allocations that maximize the utilitarian social welfare are not, in general, EF$1$. Moreover, there are additive scenarios in which the property EF$1$ fails for every allocation which maximizes the utilitarian social welfare (see Example~\ref{exp:new-and-very-simple}). Indeed, in this example we can see that  Nash and  utilitarian social welfares are independent.

{Finding an allocation that maximizes the utilitarian social welfare and satisfies the EF$1$ criteria is an NP-complete problem when the number of agents is greater than or equal to three (see \cite{aziz2021computing}).}
However, for some specialized scenarios, there are very simple algorithms in polynomial time which find EF$1$, and  even 
{EFX }
allocations. That is the case of identical utilities (see Barman et al. \cite{Barman2018}) \myred{and others like  bivaluated utilities (Ebadian et al. \cite{EPS22}) or  binary utilities \cite{Barman2018} (see Section~\ref{sec:related_work})}. 

In this work, we consider a class of additive functions, called \myred{generalized binary} utility functions. This class is  more general than the classes of \myred{binary and} identical utilities. 
\myblue{Intuitively, each resource has a market-price; each agent either does not want the resource at all, or wants it and values it by its market price.}
In the framework of  these utility functions, the following results are established:

\begin{enumerate}
    \item A characterization for Pareto optimal allocations (Theorem \ref{thm-main-1}).

    \item Each allocation that maximizes the Nash social welfare also maximizes the utilitarian social welfare (Theorem \ref{coro-main-thm}).

    \item Constructive proofs of the existence of  allocations which maximize the utilitarian social welfare, which are PO and respectively EF$1$ and EFX
    \myred{ for goods and for chores}. These allocations are obtained in  polynomial time (Theorems~\ref{thm:main} and~\ref{thm:main-efx}).
\end{enumerate} 

\noindent Moreover, we propose a basic algorithm in $O(nm)$ which finds, under additive utility functions,  an allocation that maximizes 
the utilitarian social welfare (Theorem~\ref{main-thm-6}).


The rest of this work is organized as follows. Section~\ref{sect-prelim} introduces the concepts and problems studied in this paper. Section~\ref{caracterization-of-PO} is devoted to a characterization of Pareto Optimality in our particular scenarios.
In Section~\ref{sect-methodology}, we propose a very simple and tractable  algorithm for maximizing the utilitarian social welfare and we study its justice properties in the case of \myred{generalized binary}  utilities.
In Section~\ref{sect-e-buyer}, we give a slight generalization of our \myred{generalized binary} scenarios and prove that most of the results obtained in the previous sections don't hold for this new class of scenarios.
\myred{Section~\ref{sec:related_work} contains a comparison of our work with other works using specialized scenarios for tackling the problems of fairness and efficiency. }
We conclude in Section~\ref{sect-conclusion} with some final remarks. The proofs\myred{, our detailed algorithm and some detailed examples} can be found in~\ref{proofs}.

\section{Preliminaries}\label{sect-prelim}

The set of agents is denoted by $\MA = \left\{1, \dots ,n \right\}$  and the set of resources  is denoted by $\MR = \left\{r_1,\dots,r_m \right\}$. So, $\vert \MR \vert=m$ and $\vert \MA \vert=n$.
An allocation of resources is a function  $F: \mathcal{\MR}\longrightarrow\mathcal{\MA}$. 
For each  agent $i$,  $F^{-1}(i)=\{r\in\MR: F(r)=i\}$ is the set of resources (or bundle) assigned to $i$. 
The set of all possible allocations is denoted by $\MA^{\MR}$. The number of possible allocations depends on $n$ and $m$, given that  
$\vert \MA^{\MR} \vert = \vert \MA \vert ^{\vert \MR \vert}=n^m$. 
The set of all subsets of $\MR$ is denoted by $\partes
\MR$. \myred{We consider two kinds of resources: {\em goods} and {\em chores}.
Goods are resources that agents are supposed to accept and chores are resources that agents are supposed to reject. We will assume that  either all resources are  goods or all resources are chores.
}

The preference of agents over resources is established through utility functions \myred{which are additive,  that is, functions of the type 
  $u:\partes
\MR\rightarrow \mathbb{R}$ which satisfy:
\begin{itemize}
\item $u(\emptyset)= 0$; 
\item $\forall S\in\partes \MR$ with $S\neq \emptyset$,  $u(S)=\displaystyle \sum_{s\in S}u(\{s\})$.
\end{itemize}
When the resources are goods, we always have $\forall S\in\partes \MR$, $u(S)\geq 0$.
When the resources are chores, we always have $\forall S\in\partes \MR$, $u(S)\leq 0$.
}
For each $i \in \MA$, $u_i$ denotes the additive  utility function associated to $i$.   
For simplicity, $u(\{s\})$ will be denoted by $u(s)$.  If every agent establishes an additive utility function in a problem of indivisible resource allocation, we say that it is an additive scenario.  

\begin{defi}\label{def-buyer-utility function}
\myred{Assume that for each $r_k\in\MR$ we have a mapping  $r_k\mapsto p_k$ where $p_k$ is a real number different from $0$.}
Let $u$ be  a function  $u:\partes
\MR\rightarrow \mathbb{R}$. We say that
 $u$ is  a \myred{generalized binary (g-binary for short)}  utility function if:
\begin{itemize}
    \item $u$ is an additive utility function;
    \item for every $ r_k\in\MR$, 
    $u(r_k)\in \{0,p_k\}$.
\end{itemize}
\end{defi}

We say that we are in a \myred{goods (chores) g-binary
} scenario \myred{when all the resources are goods (chores) and} every agent of the resource allocation problem   \myred{has a  utility function which is g-binary for the same mapping $r_k\mapsto p_k$.}

\myblue{Note that in case of goods (chores) $p_k$ can be viewed as the market-price (rejection intensity) of the resource $r_k$. The value 0 does not change an agent's utility; in case of goods, this value can be interpreted as a reject, whereas in chores, it can be interpreted as a preference.} 

Clearly, if $u$ is an additive utility function such that for each $r_k\in \MR$, $u(r_k)\in \{0,1\}$, then  $u$ is a \myred{g-binary} utility function. This type of  function is known as binary utility function, see \cite{amanatidis2020}.
Thus, the \myred{goods g-binary} scenarios are a generalization  of binary scenarios.

\myred{We assume that in goods and chores scenarios there is not a resource $r$ such that for every agent $i\in \MA$, $u_i(r)=0$. Therefore, in a goods g-binary scenario, $\forall r_j\in R$, $\max\{u_i(r_j): i\in A\}=p_{r_j}>0$ and $\min\{u_i(r):i\in A\}=0$. While, in a chores g-binary scenario, $\forall r_j\in R$, $\max\{u_i(r_j): i\in A\}=0$ and $\min\{u_i(r):i\in A\}=p_{r_j}<0$.}

A \myred{goods g-binary} scenario in which the value $0$ is not taken by the utility functions is called a scenario of identical utility functions. These scenarios were considered by Barman et al. \cite{Barman2018}.
{There are examples where imposing a scenario of identical utilities is not very adequate to find an allocation that best satisfies the  agents. The following example shows such a situation.

\begin{ejem}
Suppose there are two goods $r_1, r_2$  and two agents $1, 2$. Agent $1$ wants $r_1$ and \myred{his utility for this good is} $p_1$ , that is $u_1(r_1)=p_1$. Agent  $1$ does not want $r_2$ and \myred{his utility for this good is null}, that is $u_1(r_2)=0$. For agent $2$, it is all the contrary,  more precisely: $2$ wants $r_2$ and \myred{his utility for this good is} $p_2$, that is $u_2(r_2)=p_2$ but $2$ does not want $r_1$ and \myred{his utility for this good is null}, that is $u_2(r_1)=0$. In such a \myred{g-binary} scenario, it is clear that the best allocation corresponds to allocating $r_1$ to $1$ and $r_2$ to $2$. In a scenario with identical utilities the agents can't express that they are not interested in a good (utility $0$). Thus the allocation which 
gives $r_1$ to $2$ and $r_2$ to $1$, in a scenario of identical utilities produces the same social utilitarian welfare but this allocation is far from satisfying the agents.
\end{ejem}
}

Note that the \myred{goods g-binary} scenarios are a simplification of the Fisher market  model in economy
\cite{BS00}.

\subsection{Fairness, efficiency, and social welfare}\label{subsection:F-E-SW}
An attractive fairness criterion in additive scenarios is the absence of envy. If no agent strictly prefers the bundle assigned to another agent instead of its own bundle\myred{, the allocation is envy-free. More precisely, an allocation  $F$ is   {\bf envy-free}     (EF) if
 $\forall i, j \in \MA,\quad u_i(F^{-1}(i)) \geq u_i(F^{-1}(j))$. 
}
If there exists an agent $i\in \MA$ such that $u_i(F^{-1}(i)) < u_i(F^{-1}(j))$ for some $j\in\MA $, then the agent $i$ envies the agent $j$. The property EF is the most desirable property, but in a simple example of an indivisible resource with two agents it is impossible to find an allocation without envy.
In the literature \cite{Budish2011, caragiannis2019,amanatidis2020,EPS21,EPS22}, weaker versions of the envy free property can be found.
The weakest among these  is  envy free up to one good.  
The following definition establishes the main weak  envy free notions \myred{for goods and chores.\footnote{\myred{Actually, one can give a unified and more compact definition but it is clearer if we split it into two cases: goods and chores.}}}

\myred{
\begin{defi}\label{def-EF1}
Let $F$ be in $\MA^{\MR}$,  

\begin{enumerate}
    \item $F$ is an {\bf envy-free up to one resource}  (EF$1$) allocation for {\bf goods}  if $\forall i, j \in \MA$,  $ \exists g\in F^{-1}(j)$ such that
\begin{equation}
      u_i(F^{-1}(i))\geq u_i(F^{-1}(j)\setminus \{g\} )
\end{equation}
\item $F$ is an {\bf envy-free up to one resource}  (EF$1$) allocation for {\bf chores}  if $\forall i, j \in \MA$,  $ \exists g\in F^{-1}(i)$ such that
\begin{equation}
      u_i(F^{-1}(i)\setminus \{g\})\geq u_i(F^{-1}(j) )
\end{equation}

\item $F$ is an  {\bf envy-free up to any non zero valued resource} (EFX) allocation for {\bf goods} if $\forall i, j \in \MA,$
\begin{equation}
     \forall g\in F^{-1}(j)\,\mbox{with }
    u_i(g)>0,\quad u_i(F^{-1}(i))\geq u_i(F^{-1}(j)\setminus \{g\} )
\end{equation}

\item $F$ is an  {\bf envy-free up to any non zero valued resource} (EFX) allocation for {\bf chores} if $\forall i, j \in \MA,$
\begin{equation}
     \forall g\in F^{-1}(i)\,\mbox{with }
    u_i(g)<0,\quad u_i(F^{-1}(i)\setminus \{g\})\geq u_i(F^{-1}(j) )
\end{equation}
\item $F$ is an  {\bf envy-free up to any  valued resource  } (EFX$_0$) allocation for {\bf goods} if $\forall i, j \in \MA,$
\begin{equation}
     \forall g\in F^{-1}(j), \, 
     \quad u_i(F^{-1}(i))\geq u_i(F^{-1}(j)\setminus \{g\} )
\end{equation}

\item $F$ is an  {\bf envy-free up to any  valued resource  } (EFX$_0$) allocation for {\bf chores} if $\forall i, j \in \MA,$
\begin{equation}
     \forall g\in F^{-1}(i),\, 
     \quad u_i(F^{-1}(i)\setminus \{g\})\geq u_i(F^{-1}(j) )
\end{equation}

\end{enumerate}

\end{defi}
}

It is clear that in the case of additive utilities
\myred{for goods or for chores} we have:
\[
EF\Rightarrow EFX_0 \Rightarrow EFX \Rightarrow EF1
\]

The efficiency, also known as Pareto efficiency or Pareto optimality, aims at characterizing when the allocation best satisfies the group. \myred{Remember that if
   $F$ and $G$ are allocations in $\MA^{\MR}$, 
we say that $F$ is  {\em  Pareto dominated} by   $G$, when:
    \begin{itemize}
    \item $\forall i\in \MA,\, u_i(F^{-1}(i)) \leq u_i(G^{-1}(i))\,\, \mbox{and}$
\item $      \exists  j\in \MA\,\mbox{ such that }\, u_j(F^{-1}(j)) < u_j(G^{-1}(j))  
$
    \end{itemize}
We say that  $G$ is {\em Pareto optimal} (PO),  if it is not Pareto dominated by another allocation.
} 

One way to measure the social satisfaction of the agents is through the Nash and utilitarian social welfare functions. \myred{Let us recall their precise definitions.
 The utilitarian social welfare of  $F$, denoted by $SW_u(F)$, is defined by  
\begin{equation}
SW_u(F) = \sum_{i\in \mathcal{A}} u_i(F^{-1}(i))
\end{equation}
 we put $MSW_u=\{F: SW_u(F)\geq SW_u(G),\,\forall G\in \MA^{\MR}\}$.
 The  Nash social welfare, denoted by $SW_{Nash}(F)$, is  defined by 
\begin{equation}
    SW_{Nash}(F) = \prod_{i\in \MA} u_i(F^{-1}(i))
\end{equation}
we put 
$MSW_{Nash}=\{F: SW_{Nash}(F)\geq SW_{Nash}(G),\,\forall G\in \MA^{\MR}\}$.
}

\myred{
A well-known result is that any allocation that maximizes utilitarian social welfare is PO: 

\begin{theorem}\label{lema-SWU-impli-OP}
Under an additive scenario, let $F$ be  in $\MA^{\MR}$,   if $F\in MSW_u$, then $F$ is PO.
\end{theorem}
}

\myred{Although in general it is not always possible to find EF allocations,} fortunately,  Caragiannis and colleagues \cite{caragiannis2019} showed that under an additive scenario \myred{for goods} it is possible to find EF1 and PO resource allocations. Actually, they prove the following  theorem:

\begin{theorem}\label{thm-Caragiannis}[Caragianis et al. \cite{caragiannis2019}]
Under an additive scenario \myred{for goods}, every allocation\footnote{Actually, the Theorem as stated is true when the maximum Nash social welfare is strictly positive. When it is zero, it is necessary to impose that the set of agents having a good is a maximal set (see \cite{caragiannis2019}).} that maximizes Nash social welfare is EF1 and PO.
\end{theorem}

\myred{
Note that, with the help of Theorem~\ref{thm-Caragiannis}, it is easy to see that the converse of Theorem~\ref{lema-SWU-impli-OP} is false. 
The following example shows this.

\begin{ejem}\label{exp:new-and-very-simple} 
Let $\MR = \{r_1,r_2\}$ be the set of resources and $\MA=\{1,2\}$  the set of agents where each agent  $i$ establishes its {utility function}  $u_i$ over every resource, according to Table \ref{tab:pref_of_agents}.
\begin{table}[H]
    \begin{center}
    \caption{ Utility functions.}
    \label{tab:pref_of_agents}
    \begin{tabular}{ccc}
          \hline
          &$r_1$&$r_2$ \\ \hline
          $u_1$&10&10 \\ 
          $u_2$&3&2 \\ \hline
    \end{tabular}
    \end{center}
\end{table}
       
Let $F$ and $G$ be the allocations defined by Table \ref{tab:def_of_allocation}.   
\begin{table}[H]
    \caption{Allocations.}
    \label{tab:def_of_allocation}
    \centering
    \begin{tabular}{ccc}
        \hline
          & $r_1$&$r_2$  \\
        \hline
         $F$&1&1 \\
         $G$&2&1 \\ \hline
    \end{tabular}
\end{table}
The utility by bundle received by each agent and social welfare are showed in Table \ref{tab:utility_bundle}.
\begin{table}[H]
    \caption{Utility by bundle received and social welfare.}
    \label{tab:utility_bundle}
    \centering
    \begin{tabular}{ccccc}
    \hline
      & 1&2&$SW_u$&$SW_{Nash}$  \\
    \hline
      $u_i(F^{-1}(i))$&20&0&20&0 \\
      $u_i(G^{-1}(i))$&10&3&13&30 \\ \hline 
    \end{tabular}
\end{table}
It is easy to see that $F$  is the only allocation which maximizes utilitarian social welfare. However, $F$ is not $EF1$ because  agent 2 envies agent 1, $u_{2}(F^{-1}(2)) = 0 < 5 = u_{2}(F^{-1}(1))$,
but even after  removing $r_1$ or $r_2$ from $(F^{-1}(1))$, envy does not disappear: 
\[
\begin{array}{l}
  u_{2}(F^{-1}(2)) = 0 < 2 = u_{2}(F^{-1}(1) \backslash \{ r_1 \}) \mbox{ and }      \\
  u_{2}(F^{-1}(2)) = 0 < 3 = u_{2}(F^{-1}(1) \backslash \{ r_2 \}) 
\end{array}\]
The maximum Nash social welfare is reached at 30, 
so $G$ maximizes the Nash social welfare. By  Theorem~\ref{thm-Caragiannis} $G$ is $EF1$, 
but it does not maximize the utilitarian social welfare.  In the following table,  we identify the properties that satisfy $F$ and $G$; if an allocation satisfies a property we use  \ding{51} and    \ding{53} otherwise.
\begin{table}[H]
    \caption{Identification of properties (within a non g-binary scenario).}
    \label{tab:identify_without}
    \centering
    \begin{tabular}{ccccc}
    \hline
         & PO&EF1&$MSW_u$&$MSW_{Nash}$  \\
    \hline
         $F$&\ding{51} &\ding{53}&\ding{51}&\ding{53}\\
         $G$&\ding{51}&\ding{51}&\ding{53}&\ding{51}\\
    \hline
    \end{tabular}
  \end{table}
\end{ejem}
}


\section{\myred{A characterization of Pareto Optimality}}\label{caracterization-of-PO}

 The following theorem shows that under \myred{g-binary scenarios}, the converse of  Theorem \ref{lema-SWU-impli-OP} is satisfied; i.e., having the property PO and belonging to $MSW_u$ are equivalent.

\begin{theorem}\label{thm-main-1}
Assume  a \myred{goods (chores) g-binary} scenario and let $F$ be an allocation in $\MA^ {\MR}$. Then, $F$ is PO if, and only if, $F\in MSW_u$.
\end{theorem}

Note that, in general additive scenarios, Theorem~\ref{thm-main-1}  does not hold as Example~\ref{exp:new-and-very-simple} reveals.

We have seen in Example~\ref{exp:new-and-very-simple} that  $MSW_{Nash}\not\subseteq MSW_u$. However,
under \myred{goods g-binary} scenarios, Theorems \ref{thm-Caragiannis} and \ref{thm-main-1} 
together, tell us that  $MSW_{Nash}\subseteq MSW_u$. This is important and will be stated in the following result.

\begin{theorem}\label{coro-main-thm}
Under a \myred{goods g-binary} scenario, each allocation that maximizes Nash social welfare also maximizes utilitarian social welfare.
\end{theorem}

\myred{A straightforward corollary of  Theorem~\ref{coro-main-thm} and Theorem~\ref{thm-Caragiannis}, is the existence, under goods g-binary scenarios, of allocations  maximizing the utilitarian social welfare and satisfying the EF1 property. The existence is based in finding an allocation producing a maximum Nash welfare, a hard problem from a computational point of view. In the next section we will see that in g-binary scenarios it is easy to compute allocations EF1 producing a maximal utilitarian welfare. }

\section{\myred{A simple Algorithm and its behavior  in some additive scenarios}}\label{sect-methodology}
\label{sec:BuildingTheAllocation}

\myred{In this section, we propose a very natural and simple algorithm  and we analyze its behavior in additive scenarios, in particular in g-binary scenarios. }

\myred{Let us start with a key result  in the conception of the algorithm. It} establishes that  an allocation distributes resources to the agents that maximize them if, and only if, this allocation maximizes the utilitarian social welfare. 
One can find a straightforward proof of this result\footnote{Actually, in \cite{Camacho2020} the {\em if} is proved. The {\em only if} is obvious. }, using a matrix approach, in \cite{Camacho2020}.
\begin{theorem}\label{thm:utili-imp-MSWU}

Under an additive scenario, $F\in MSW_u$ if, and only if, $ \forall i \in \MA, \forall r \in F^{-1}(i)$ we have that  $u_i(r)=\max\{u_j(r):j\in \MA\}$. 
\end{theorem}

\myred{Suppose that in an additive scenario, $\MR=\{r_1,\dots,r_m\}$ (all the resources are goods or all the resources are chores) and $\MA=\Set{1,\dots,n}$.  Let $\alpha_1,\dots,\alpha_m$ be the real numbers defined in the following way: for each $r_k\in\MR$, $\alpha_k=\max\{u_j(r_k): j\in\MA\}$.}
\myred{Then, by the previous theorem the following simple algorithm defines an allocation, $\Gamma$, having a maximal utilitarian welfare, that is, $\Gamma\in MSW_u$.}

\begin{algorithm}
\myred{
\DontPrintSemicolon 
\KwIn{Two finite sets, $\MR=\{r_1,\dots,r_m\}$ for the resources, $\MA=\Set{1,\dots,n}$ for the agents,
and their respective utilities $u_i$}
\KwOut{The allocation $\Gamma$ }
$v_0 \gets (0,\dots,0)$\;
\For{$k \gets 1$ \textbf{to} $m$} {
  $\alpha_k \gets \max\{u_j(r_k): j\in\MA\}$\;
  $P_k \gets \{j\in \MA: u_j(r_k)=\alpha_k\}$\;
  $l_{k} \gets \min \left\{ \vert [v_{k-1}]_{j} \vert : j\in P_k \right\}$\;
  $M_k \gets \left\{ i \in P_k : \vert [v_{k-1}]_{i} \vert = l_{k} \right\}$\;
  $j_k \gets \min\{ M_k \}$\;
  $\Gamma(r_k) \gets j_k$ \;
  $\left[ v_{k}\right]_{i} \gets
  \begin{cases}
    \left[v_{k-1} \right]_{i} + u_{j_k}(r_k),\,&\mbox{if}\,\,i=j_k\\
    \left [v_{k-1} \right]_{i}, \,&\mbox{if} \,\, \ i \not= j_k.
  \end{cases}$  for all $i\in \MA$\; 
}
\Return{$\Gamma$}\;
\caption{An allocation for maximal utilitarian welfare}
\label{algo:allocate}
}
\end{algorithm}

\myred{The idea of Algorithm \ref{algo:allocate} is very simple: the resource $r_k$ is allocated to an agent who maximizes its utility and such that 
before this step, that is
until the partial allocation of resources $\Set{r_1,\dots,r_{k-1}}$ is done, he has the minimal utility in case of goods and the maximal utility in the case of chores. A detailed description of the behavior of this algorithm appears in the Appendix.

From Theorem \ref{lema-SWU-impli-OP} we have the following result:
\begin{coro}\label{coro-Gamma is PO}
Under an additive scenario,  $\Gamma$ is  $PO$.
\end{coro}}

\myred{ We saw in 
Example~\ref{exp:new-and-very-simple} that the  allocation $F$, the unique allocation that maximizes the utilitarian social welfare is not EF1 (see Table~\ref{tab:identify_without}). But this allocation is indeed the allocation $\Gamma$, thus in general, Algorithm~\ref{algo:allocate} does not produce an allocation EF1. 

The following observation summarizes some facts in g-binary scenarios. This will be proved throughout Example~\ref{exp:properties PO EF1} in~\ref{proofs}.
}

\myred{
\begin{remark}\label{remark-new-behavior}
It is important to note  that in g-binary scenarios there are allocations EF1 which are not in $MSW_u$. In those scenarios there are allocations which are 
$PO$ and  $MSW_u$ but they are  neither EF1 nor belong to $MSW_{Nash}$.
 Also in those scenarios 
$\Gamma$ is EF1, PO, it  is in  $MSW_u$ but, in general, it is not in $MSW_{Nash}$.
It can happen that there exist allocations that are 
EF1, PO, and they are in both  $MSW_u$ and in  $MSW_{Nash}$.
\end{remark}
}

\myred{We have already said that Algorithm~\ref{algo:allocate} produces a maximal utilitarian welfare. This occurs  in case of additive scenarios and due to Theorem~\ref{thm:utili-imp-MSWU} and the definition of Algorithm~\ref{algo:allocate}. The next result summarizes this and gives the complexity of the algorithm.}

\begin{theorem}\label{main-thm-6}
Under a \myred{goods (chores)} additive scenario,  $\Gamma \in MSW_u$  and it is obtained in $O(nm)$ operations. 
\end{theorem}

\myred{In addition to this, we will show next that in g-binary scenarios
for goods and chores,
the allocation $\Gamma$, given by Algorithm~\ref{algo:allocate}
is also envy-free up to one resource (good or chore).

}

\myred{
\begin{theorem}\label{thm:main}
Under a g-binary scenario for goods or chores, the allocation $\Gamma$ given by Algorithm~\ref{algo:allocate} produces a maximal utilitarian welfare, is EF$1$ and PO. Its run time is $O(mn)$. 
\end{theorem}

As a matter of fact, in a g-binary scenario for  chores, if for any resource $r$, 
\[max_{i\in \MA}\Set{u_i(r)}\neq min_{i\in \MA}\Set{u_i(r)}\] 
we have that $\Gamma$ is EF. That is because, in such a case, $\Gamma$ is an allocation in which all the agents give utility 0 to their bundles.

 Moreover, with a very slight modification of Algorithm~\ref{algo:allocate}, we will obtain another allocation, $\Gamma^\ast$, which is EFX. This modification  consists in taking one  more step:  reordering the resources.   In the case of goods we order the resources in decreasing  order  and in the case of chores we order the resources in increasing order   according to the maximal absolute value of utilities given to resources. More precisely, 
  a resource $r$ precedes another resource $r'$ if 
$max_{i\in \MA}\Set{\lvert u_i(r)\rvert}\geq max_{i\in \MA}\Set{\lvert u_i(r')\rvert}$.
For this $\Gamma^*$ we have the following result.

\begin{theorem}\label{thm:main-efx}
Under a g-binary scenario for goods or chores, the allocation $\Gamma^*$ given by Algorithm~\ref{algo:allocate} modified as previously indicated maximizes the utilitarian social welfare, is  EFX and PO. Moreover, 
its run time is $O(m\log m + mn)$.
\end{theorem}

It is interesting to observe that the algorithm producing $\Gamma^*$ is quite similar to Barman et al. algorithm \cite{Barman2018} producing an allocation EFX, in the case of identical scenarios for goods. In fact, in identical scenarios, the notions EFX and EFX$_0$ coincide. This is not the case in g-binary scenarios. 
This can be viewed by building  a g-binary scenario,
in which $\Gamma^*$, the allocation given by the  modified Algorithm~\ref{algo:allocate}, which is EFX, is not EFX$_0$ (see Example~\ref{ej-EFX-not-EFXO} in \ref{proofs}).

Note that if $F$ is an   allocation $EFX$ for goods (for chores) such that for all $i,j\in\MA$ such that $u_i({F}^{-1}(i))<u_i({F}^{-1}(j))$ we have that   
${F}^{-1}(j)\subseteq \{r\in \MR: u_i(r)>0 \}$ (${F}^{-1}(i)\subseteq \{r\in \MR: u_i(r)<0 \}$), then $F$ is $EFX_0$ for goods (for chores).
}

\section{\myred{Some limits of Algorithm~\ref{algo:allocate}} }\label{sect-e-buyer}

\myred{
It is natural to ask if Algorithm~\ref{algo:allocate} has other interesting properties.
For instance, does it compute a maximal egalitarian welfare\myred{\footnote{\myred{The egalitarian social welfare of an allocation $F$, denoted by $SW_e(F)$ is defined by $SW_e(F)=min\Set{u_i(F^{-1}(i)): i\in\MA}$.}}} allocation? We will see in Example~\ref{no-egalitarian} that it is not the case even when we work in g-binary scenarios.

Another question concerns the properties of Algorithm~\ref{algo:allocate} in a slight modification of g-binary scenarios. Remember that in}
 g-binary scenarios, every Pareto efficient allocation indeed maximizes utilitarian social welfare (see Theorem \ref{thm-main-1}), and as a consequence, \myred{in the case of goods,} every assignment that maximizes Nash's social welfare also maximizes the utilitarian one (see Theorem \ref{thm:utili-imp-MSWU}). \myred{Thus,  natural questions are:  in these modified scenarios,} is it possible to preserve these properties? Furthermore, does \myred{ Algorithm~\ref{algo:allocate}} achieve fair allocations?

\myred{
The following example shows that neither $\Gamma$ nor $\Gamma^*$  maximize the egalitarian social welfare.
}

\myred{ 
\begin{ejem}\label{no-egalitarian}
There are three agents and three resources. The utilities are given in Table~\ref{ejem:noegalitarianutilities}. In Table~\ref{ejem:noegalitarianallocations}
are the allocations. In Table~\ref{ejem:noegalitarianswe} appear the utilities of each agent for every allocation and the egalitarian social welfare denoted by $SW_e$.

\begin{table}[H]
    \begin{center}
    \caption{ Utility functions.}\label{ejem:noegalitarianutilities}
        \begin{tabular}{cccc}
          \hline
          &$r_1$&$r_2$&$r_3$\\ \hline
          $u_1$&4&1&2 \\ 
          $u_2$&4&0&2 \\ 
          $u_3$&4&0&0\\ \hline
    \end{tabular}
    \end{center}
\end{table}
\begin{table}[H]
    \caption{Allocations.}\label{ejem:noegalitarianallocations}
        \centering
    \begin{tabular}{cccc}
        \hline
          & $r_1$&$r_2$&$r_3$  \\
        \hline
         $\Gamma$&1&1&2 \\
         $\Gamma^*$&1&1&2 \\ 
         $A$&3&1&2\\ \hline
    \end{tabular}
\end{table}
\begin{table}[H]
    \caption{Utility by bundle received and egalitarian social welfare.}\label{ejem:noegalitarianswe}
    \centering
    \begin{tabular}{ccccc}
    \hline
      & 1&2&3&$SW_{e}$  \\
    \hline
      $u_i(\Gamma^{-1}(i))$&5&2&0&0 \\
      $u_i({\Gamma^*}^{-1}(i))$&5&2&0&0 \\
      $u_i(A^{-1}(i))$&1&2&4&1 \\\hline 
    \end{tabular}
\end{table}
 Table~\ref{ejem:noegalitarianswe} shows that $\Gamma$ and $\Gamma^*$ do not maximize the egalitarian social welfare because the allocation $A$ has egalitarian social welfare which is better than the egalitarian social welfare of $\Gamma$ and $\Gamma^*$.
\end{ejem}
}

\myred{Now,} we propose a scenario called $\epsilon$-g-binary, which slightly changes the range of the utilities of each agent, generalizing the notion of g-binary scenarios. Let us see the definition below.

\begin{defi}\label{def-alpha-buyer-utility function} 
\myred{Assume that for each $r_k\in\MR$ we have a mapping  $r_k\mapsto p_k$ where $p_k$ is a real number.}
Let $u$ be  a function  $u:\partes
\MR\rightarrow \mathbb{R}$. We say that
 $u$ is  an $\epsilon$-g-binary  utility function \myred{for goods (chores)}  if:
\begin{itemize}
\item \myred{$\epsilon\geq 0 $ ($\epsilon\leq 0 $)};
    \item $u$ is an additive utility function;
    \item for every $ r_k\in\MR$, 
    $p_k>\epsilon$  
    \myred{($p_k<\epsilon$)} and  $u(r_k)\in \{\epsilon,p_k\}$.
\end{itemize}
{When every agent has an $\epsilon$-g-binary utility function \myred{for goods (chores) with the same mapping $r_k\mapsto p_k$}, we say that the scenario is \myred{goods (chores)} $\epsilon$-g-binary.}
\end{defi}

The above definition generalizes Definition \ref{def-buyer-utility function}; \myred{in the case goods (chores)}, $\epsilon$ is the minimum \myred{(maximum)} value for the scenario. 
It means that the minimum (maximum) value could be equal to zero or greater than \myred{(less than)}  zero. {When $\epsilon$ is equal to 0, we get exactly the \myred{goods (chores) g-binary} scenario.}

In the following example, {we give  an $\epsilon$-\myred{g-binary} scenario, with $\epsilon > 0$, for which there exists an allocation $G \notin MSW_u$ having a greater Nash social welfare than all the allocations that maximize the utilitarian social welfare.}

\begin{ejem} We consider 
$n=m=3$. Table \ref{tab:ejemplo-epsilon-buyer-Nash-notimplies-U} defines our $\epsilon$-\myred{g-binary} scenario  with $b > a > \epsilon > 0$.
\begin{table}[H]
     \caption{Utility functions.}
     \label{tab:ejemplo-epsilon-buyer-Nash-notimplies-U}
     \centering
    \begin{tabular}{cccccc}
          \hline
                 & $r_1$ & $r_2$    & $r_3$    \\ \hline
          $u_1$  & $b$     & $a$        & $a$        \\ 
          $u_2$  & $b$     & $\epsilon$ & $\epsilon$ \\ 
          $u_3$  & $b$     & $\epsilon$ & $\epsilon$ \\ \hline
    \end{tabular}
\end{table}
Note that any allocation that maximizes the utilitarian social welfare assigns the resources $r_2$ and $r_3$ to  agent $1$ (see Theorem \ref{thm:utili-imp-MSWU}). {Thus, there are only three allocations, $F_1, F_2$ and $F_3$, which maximize the utilitarian social welfare. They are shown in table \ref{tab:allUtilitarian}.  This table shows also another allocation, namely the allocation $G$ that assigns the resource $r_3$ to agent $3$, who gives it a utility of $\epsilon$.} 


\begin{table}[H]
     \caption{Definition of allocations.}
     \label{tab:Asignación-ordenado-epsilon-buyerNash-notimplies-U}
     \centering
    \begin{tabular}{cccccc}
    \hline
         & $r_1$ &$r_2$ & $r_3$  \\ \hline
    $F_1$ & 1     & 1    & 1 \\
    $F_2$ & 2     & 1    & 1 \\ 
    $F_3$ & 3     & 1    & 1 \\ 
    $G$   & 2     & 1    & 3 \\ \hline
    \end{tabular}
    \label{tab:allUtilitarian}
\end{table}

Table \ref{tab:UtilityEpsilonBuyer} describes the valuations that the agents give to the received bundle in the different allocations of Table \ref{tab:allUtilitarian}; also, it shows the Utilitarian and Nash social welfare of the allocations.

\begin{table}[H]
    \caption{Utility by received bundle and social welfare.}
    \label{utilidad-recivida-ordenado-epsilon-buyerNash-notimplies-U}
    \centering
    \begin{tabular}{ccccccccc}
      \hline
                             & $1$ & $2$ & $3$ & $SW_{u}$ & $SW_{Nash}$ \\ \hline
      $u_{i}({F_1}^{-1}(i))$ & $2a+b$ & $0$ & $0$ & $2a+b$ & $0$ \\ 
      $u_{i}({F_2}^{-1}(i))$ & $2a$ & $b$ & $0$ & $2a+b$ & $0$ \\ 
      $u_{i}({F_3}^{-1}(i))$ & $2a$ & $0$ & $b$ & $2a+b$ & $0$ \\ 
      $u_{i}({G}^{-1}(i))$ & $a$ & $b$ & $\epsilon$ & $a+b+\epsilon$ & $ab\epsilon$ \\ \hline
    \end{tabular}
    \label{tab:UtilityEpsilonBuyer}
\end{table}

Note that  allocation $G$ has better Nash social welfare than $F_1, F_2$ and $F_3$. {This is because the Nash social welfare of $F_1, F_2$ and $F_3$ is 0 and the Nash social welfare of $G$ is strictly positive. } Then $MSW_u \cap MSW_ {Nash} = \emptyset$. Consequently, there are PO allocations that are not contained in $MSW_u$; and therefore,   Theorem \ref{thm-main-1}, instantiated in these scenarios, does not hold.

{\myred{Note that in this example $\Gamma=\Gamma^*=F_1$. But,} regarding fairness, 
in allocations  $F_1$, $F_2$, and $F_3$, at least one agent does not receive any good, and at the same time, this agent values positively the two resources received by another agent. Then, it is not possible to eliminate the envy  removing only one resource.}
Thus, in this scenario, the allocations which maximize  utilitarian social welfare do not {necessarily} satisfy fairness.
In particular,   Algorithm~\ref{algo:allocate} does not attain success in finding an allocation $EF1$.
\end{ejem}

\section{\myred{Related work}}
\label{sec:related_work}

\myred{
Due to the fact that allocations producing a maximal utilitarian welfare are, in additives scenarios, Pareto optimal, in order to assure the Pareto optimality and EF1, it is natural to look for allocations which are in $MSW_{u}$ and are simultaneously  EF1.
However, in general additive scenarios, the problem of finding allocations producing a maximal utilitarian social welfare being simultaneously EF1 is NP-hard \cite{aziz2021computing}. Thus, 
several works adopt the strategy of working in scenarios  in which the utility functions are restricted. Next, we resume the main proposals studied in the literature, their results and the connections with our framework.

Actually, the work of Caragiannis et al. \cite{caragiannis2019}, under general goods  additive scenarios shows that some special allocations producing maximum Nash social welfare are EF$1$ and PO. However, finding this kind of allocations is in general an NP-complete problem. Moreover, this does not solve the problem of finding an allocation that produces a maximal utilitarian social welfare and that is also EF1, because the allocations in $MSW_{Nash}$ are not necessarily in $MSW_{u}$. However, in g-binary scenarios,
we have $MSW_{Nash}\subseteq MSW_{u}$, because in these scenarios being Pareto optimal is equivalent to producing a maximal utilitarian welfare.

In the same vein, Amanatidis et al. \cite{amanatidis2020}  
connect the maximum of Nash and EFX. 
They consider $k$-valued scenarios, that is, scenarios where there is a set of $k$ distinct possible values, all real numbers, that agents can assign to their utilities.
For 2-valued scenarios they prove that the allocations maximizing the Nash social welfare are EFX$_0$. They also prove that for $k$-valuated scenarios, with $k\geq 3$, it is not true that the allocations   maximizing the Nash social welfare are EFX$_0$.
To the best of our knowledge, there is no known simultaneous existence of EFX and  PO allocations under $k$-valued scenarios. Note that these scenarios are a generalization of $\epsilon$-g-binary scenarios.

Actually, the 2-valued scenarios, called also bivaluated,  are deeply studied by Ebadian et al. \cite{EPS22}. They prove that for the case of chores, in that scenario, there exist polynomial time algorithms for calculating allocations which are simultaneously EF1 and PO. They studied also other fairness properties  and other scenarios for which  they prove the existence of allocations in polynomial time which are fair and PO.
Note that the bivaluated scenario and the g-binary are independent.

A  particular case of a bivaluated scenario is a binary scenario, where the utility functions can take only the values zero or one. In those scenarios, Barman et al.~\cite{Barman2018} proposed an algorithm to find an allocation that satisfies EFX and PO constraints and runs in polynomial time. Actually, they propose an algorithm that finds an allocation maximizing the Nash social welfare in polynomial time in the case of binary utilities. In the case of identical utilities, they propose also an algorithm in polynomial time to compute an allocation EFX which maximizes the utilitarian social welfare. They show also that this allocation is a good approximation of a maximum of Nash social welfare.

In Figure~\ref{img:related_work}, we summarize the results of this discussion.
The left branch in the figure  goes from the most specialized scenario, the binary scenario, to the most general additive scenario passing through the bivaluated an $k$-valuated scenarios. The right branch  goes from identical 
scenario, the most specialized in this branch, to additive scenarios (the most general) passing through the new scenarios we study in this work:
g-binary and $\epsilon$-g-binary scenarios. 

We have noted that g-binary scenarios generalize simultaneously binary and identical scenarios and it is easy to see that $\epsilon$-g-binary scenarios generalize bivaluated scenarios.

}

\begin{figure}[H]
\begin{center}
    \input{diagram4}
\end{center}
\caption{The  boxes having a bold face framework represent the classes and results proposed in this work. The arrows show strict inclusion. Check marks $(\checkmark)$ denote the existence of fairness (EFX) and efficiency (PO) properties, while $(?)$ denotes an open question. }
\label{img:related_work}
\end{figure}

\section{Concluding remarks}\label{sect-conclusion}

In this work, we have studied fair and efficient allocations for indivisible resources \myred{(goods and chores)} when  the valuations are defined through g-binary utility functions. We showed that there exists an allocation that is envy-free up to one item (EF1) \myred{(for goods and chores)} and maximizes utilitarian social welfare ($MSW_u$) using g-binary utility functions. This result can  be slightly modified to produce an allocation that is envy-free up to any \myred{item valuated different from zero} (EFX) \myred{(for goods and chores)}.
In this framework 
a characterization of Pareto optimality is given.

Actually, we build  algorithms for finding  allocations which are  EF1,   EFX and  simultaneously $MSW_u$ \myred{(for goods and chores)}. They run in  polynomial-time. However, the algorithm finding an allocation EF1 is computationally slightly better than  the algorithm finding an allocation EFX.

{
The results of Section~\ref{sect-e-buyer} show, in some way, that the g-binary scenarios are a sort of limiting scenario in which it is easy to provide EF1 and PO  properties related to utilitarian social welfare. This is achieved  by considering when a resource   \myred{is indifferent to} an agent. However, a little bit beyond that, the results do not hold.}

\myred{It is known that finding allocations producing a maximum of Nash is NP-hard even in identical scenarios (see \cite{Barman2018}). Thus, in g-binary scenarios this problem will  necessarily be  NP-hard. Then, the natural question is if we can adapt the techniques of Barman et al. \cite{Barman2018} in order to prove that Algorithm~\ref{algo:allocate} produces a good approximation of a maximum of Nash social welfare in the case of  g-binary scenarios.} 

\section*{Acknowledgements}
{The first, second and fourth authors   thank the   Vice Chancellery  of Research and Innovation of Yachay Tech University which has partially funded this work through the  {\em Social Welfare and Justice in Decision Making} project code MATH22-08.}

The third author  has benefited from the support of the AI Chair BE4musIA of the French National Research Agency (ANR-20-CHIA-0028) and has also been partially funded by the program PAUSE of Collège de France.

\myred{Thanks to the referees whose observations have greatly helped to improve the results of this work and its presentation.}
We also thank  Professor Olga Porras for her careful proof reading.

\appendix

\section{\myred{Algorithm, proofs and examples}}\label{proofs}

A more detailed explanation of Algorithm~\ref{algo:allocate} is as follows:

\begin{enumerate}
    \item Let $v_0 =(0,\dots,0)$  be the initial vector of partial valuations, with size $n$.
\item  \label{step-2} For $k=1$ to $m$:     
\begin{enumerate}
\item We take $r_k\in \MR$.
\item {$\alpha_k=\max\{u_j(r_k): j\in\MA\}$.}
\item  Let $P_k$ and $M_k$ be the sets  given by:
\begin{equation}\label{eq-P_k}
    P_k=\{j\in \MA: u_j(r_k)=\alpha_k\}
\end{equation} and
\[l_{k} = \min \left\{ \myred{\lvert[v_{k-1}]_{j}\rvert} : j\in P_k \right\}\]
\begin{equation}\label{eq-def-M_k}
M_k=\left\{ i \in P_k : \myred{\lvert[v_{k-1}]_{i}\rvert }= l_{k} \right\}
    \end{equation}
    where $[v_{k-1}]_{i}$ is the position $i$ of $v_{k-1}$.
\item Let $j_k$ be the minimum of $M_k$.
\item Allocation of $ r_k $:  \begin{equation}\label{eq-defi-alocation-F}
    \Gamma(r_k)=j_k
\end{equation}
\item Updating  the vector $v_k$ of partial utilities, for all $i\in \MA$:
$$\left[ v_{k}\right]_{i} =
\begin{cases}
\left[v_{k-1} \right]_{i} + u_{j_k}(r_k),\,&\mbox{if}\,\,i=j_k\\
\left [v_{k-1} \right]_{i}, \,&\mbox{if} \,\, \ i \not= j_k.
\end{cases}
$$
 \end{enumerate}
\end{enumerate}

Notice that  $P_k\neq \emptyset$. Therefore,  $M_k\neq \emptyset $ and $M_k\subseteq \mathbb{N}$; so, there exists a minimum for $M_k$.
On the other hand, $\forall i\in \MA$, $u_i$ is an additive utility function, so  each position  $i$ of $v_k$,   $\left[v_k \right]_{i}$, is the partial valuation given by agent $i$ to its assigned bundle up to step $k$.

\bigskip

\noindent{\bf Theorem 3.}
{\em Assume  a \myred{goods (chores) g-binary} scenario and let $F$ be an allocation in $\MA^ {\MR}$. Then, $F$ is PO if, and only if, $F\in MSW_u$.
}

\medskip

\begin{myproof}
\myred{({\em only if} part)} Let  $F\in \MA^{\MR}$, and we suppose that $F$ does not maximize utilitarian social welfare. We want to show that $F$ is not Pareto optimal.

Since all agents consider g-binary utility functions and $F$ does not maximize utilitarian social welfare, by Theorem \ref{thm:utili-imp-MSWU}, there exist $r^*\in \MR$ and $j\in\MA$ such that $\myred{u_j(r^*)=\max\{0,p^*\}}$ and
$\myred{u_i(r^*)=\min\{0,p^*\}}$ with $i=F(r^*)$
and $\myred{p^*\neq 0}$.
Let  $G\in \MA^{\MR}$ be given by $$G(r)=\begin{cases}
F(r),\,\, \mbox{if} \,\, r\neq r^*\\
j,\,\, \mbox{if} \,\, r=r^*
\end{cases}$$
\myred{Note that $r^*\notin F^{-1}(j)$,
$G^{-1}(j) =F^{-1}(j) \cup \{r^*\}$, 
$G^{-1}(i) =F^{-1}(i) \setminus \{r^*\}$, and for all $k\in \MA$ with  $k\notin \{i,j\}$,
$F^{-1}(k)=G^{-1}(k)$.

In a good g-binary scenario, $u_j(r^*)=p*>0$ and $u_i(r^*)=0$. So, $u_j(F^{-1}(j))<u_j(G^{-1}(j))$ and $u_k(F^{-1}(k))=u_k(G^{-1}(k))$ for $k\neq j$. Hence,
$F$ is Pareto dominated by  $G$.

Now,  in a chore g-binary scenario, $u_j(r^*)=0$ and $u_i(r^*)=p^*<0$. Thus, $u_i(G^{-1}(i))=u_i(F^{-1}(i))-u_i(r^*)>u_i(F^{-1}(i))$ and $u_k(F^{-1}(k)=u_k(G^{-1}(k))$ for $k\neq j$. Therefore, $G$  dominates  $F$. 

In both cases  $F$ is not PO.

({\em  if} part) Follows from Theorem~\ref{lema-SWU-impli-OP}.
}
\end{myproof}

\medskip

The following observation is very useful in the proofs. Its proof is obvious by equations  \eqref{eq-defi-alocation-F},  \eqref{eq-P_k} and \eqref{eq-def-M_k}.

\begin{remark}\label{lemma-properties1-F} 
We assume an additive scenario. Let  $\Gamma$ be
the allocation of Algorithm~\ref{algo:allocate}.
For every $r_k\in\MR$ and every $j\in \MA$ such that   $\Gamma(r_k)=j$, we have   
\begin{enumerate}
    \item \label{i}   $\forall i(i\in P_k\implies    \myred{ \lvert\left[ v_{k-1} \right]_j\rvert \leq \lvert \left[v_{k-1} \right]_i \rvert)}$ 
    \item\label{ii} $u_j(r_k)=\alpha_k=max\{u_i(r_k): \forall i\in \MA\}$
\end{enumerate}
\end{remark}

\noindent{\bf Theorem 6.}
{\em 
Under a \myred{goods (chores)} additive scenario,  $\Gamma \in MSW_u$  and it is obtained in $O(nm)$ operations. }

\medskip

\begin{myproof}
Let $\Gamma$ be the allocation of Algorithm~\ref{algo:allocate}.  By 
Remark~\ref{lemma-properties1-F}, part  \ref{ii},  
and  Theorem~\ref{thm:utili-imp-MSWU}, we have that $\Gamma\in MSW_u$. 
The proposed algorithm starts initializing the vector of partial utilities $v_0$ with zeros, which has one position by each one of the $n$ agents. 
This step demands $O(n)$ operations. 
Step-\ref{step-2} allocates the $m$ resources finding the agent with minimum partial utility in each iteration, following \eqref{eq-def-M_k}. 
Allocating $m$ resources, finding the minimum of $n$ agents demands $O(nm)$ operations.
In the next steps, the agent receives the resource in \eqref{eq-defi-alocation-F}; for all the resources, this step and the updating of the vector of partial utilities runs in $O(m)$.
Building the resource allocation runs in $O(max(n, nm, m))$, thus, we have that the proposed algorithm is  in $O(nm)$.
\end{myproof}

\bigskip

\myred{
In order to prove Theorem~\ref{thm:main},  we  establish a technical lemma for \myred{goods or chores g-binary scenario}. To establish this lemma we need the following definitions:

\begin{align*}
    A_{ij} &= \{r \ | \ u_i(r)=\max\{0,p_r\} \, \wedge \, u_j(r)=\min\{0,p_r\}\}, \\
    B_{ij} &= \{r \ | \ u_i(r)=\min\{0,p_r\} \, \wedge \, u_j(r)=\max\{0,p_r\}\}, \\
    C_{ij} &= \{r \ | \ u_i(r)=\max\{0,p_r\}=u_j(r)\}, \quad
    \\
    D_{ij} &= \{r \ | \ u_i(r)=\min\{0,p_r\}=u_j(r)\},\\
    R_{*} &= \{r \ | \ u_i(r)=\min\{0,p_r\}, \ \forall i \in A\}.
\end{align*}

Note that for every $i,j$ we have
\begin{equation}\label{partition}
R=A_{ij}\cup B_{ij} \cup C_{ij} \cup D_{ij} 
\end{equation}
and this is a partition of $R$.

\begin{lema}\label{technicallemma}
Under a g-binary  scenario for goods or chores, if $\Gamma$ is computed via Algorithm \ref{algo:allocate}, then for all $i,j \in A$ we have that
\begin{equation}\label{eq-lema-1}
    u_{i}\left(\Gamma^{-1}(i)\right)=\left\{\begin{array}{ll}
u_{i}\left(\Gamma^{-1}(i) \cap A_{i j}\right) + u_{i}\left(\Gamma^{-1}(i) \cap C_{i j}\right) &\quad \text{in  goods g-binary scenario} \\
u_{i}\left(\Gamma^{-1}(i) \cap R_{*}\right)  &\quad \text{in  chores g-binary scenario}
\end{array}\right.
\end{equation}
and,
\begin{equation}\label{eq-lema-2}
    u_{i}\left(\Gamma^{-1}(j)\right)=\left\{\begin{array}{ll}
 u_{i}\left(\Gamma^{-1}(j) \cap C_{i j}\right) &\quad \text{in  goods g-binary scenario} \\
u_{i}\left(\Gamma^{-1}(j) \cap B_{ij}\right) + u_{i}\left(\Gamma^{-1}(j) \cap R_{*}\right)  &\quad \text{in chores g-binary scenario}
\end{array}\right.
\end{equation}

\end{lema}
\begin{myproof}
For each $r \in R$, we consider $M_r = \{k \in A: u_k(r)=\max\{0, p_r\}\}$ and $ m_r = \{k \in A: u_k(r)=\min\{0, p_r\}\}$.
Note that, $M_r \cap m_r = \emptyset$ and $M_r \cup m_r = A$. Moreover, by definition of   $\Gamma$, if $M_r\neq \emptyset$, then $\Gamma(r)\in M_r$.
\begin{claim}\label{claim-technical-Lemma} For all  $i,j \in A$,
\begin{enumerate}
    \item $A_{ij} \cap \Gamma^{-1}(j) = \emptyset=B_{ij} \cap \Gamma^{-1}(i) $.\\
    In fact,  if  $r \in A_{ij} \cap \Gamma^{-1}(j)$, then $\Gamma(r) = j \in M_r$ and $j \in m_r$ which is a contradiction. So, $A_{ij} \cap \Gamma^{-1}(j) = \emptyset$. As $B_{ij}=A_{ji}$, then $B_{ij} \cap \Gamma^{-1}(i)=\emptyset$.
    \item $D_{ij}\cap \Gamma^{-1}(i)= R_*\cap\Gamma^{-1}(i)$.\\
    Suppose that  $R_*=\emptyset$.
    Then $M_r\neq \emptyset$ for all $r$. If $r\in D_{ij}\cap \Gamma^{-1}(i) $, then $\Gamma(r) = i \in M_r$ and $i \in m_r$ which is a contradiction. So, $D_{ij}\cap \Gamma^{-1}(i)=\emptyset=R_*\cap\Gamma^{-1}(i)$. Now, suppose that  $R_*\neq \emptyset$, 
    if $r\in R_*$, then $m_r\neq \emptyset$ and, by definition of $\Gamma $, $\Gamma(r)\in m_r$. Given  $r\in D_{ij}\cap\Gamma^{-1}(i)$, then $u_i(r)=\min\{0,p_r\}=u_{j}(r)$ and $\Gamma(r)=i$. Therefore, for all $k\in A$, $u_k(r)=\min\{o,p_r\}$; otherwise, $\Gamma(r)\notin m_r$. Thus, $r\in R_*\cap\Gamma^{-1}(i) $. But, $R_*\subseteq D_{ij}$. Thus, $D_{ij}\cap \Gamma^{-1}(i)= R_*\cap\Gamma^{-1}(i).$
   \item \label{claim1-item-3}$\Gamma^{-1}(i)=(A_{ij}\cap \Gamma^{-1}(i))\cup (C_{ij}\cap \Gamma^{-1}(i))\cup (R_*\cap \Gamma^{-1}(i))$.\\
   This is due to parts 1 and 2 of Claim\ref{claim-technical-Lemma} and Equation~\ref{partition}. 
\end{enumerate}
\end{claim}
Now we prove the lemma. By Claim\ref{claim-technical-Lemma} part \ref{claim1-item-3} and the fact that $B_{ij}=A_{ji}$ and $C_{ij}=C_{ji}$, we have
\begin{eqnarray*}
    u_i(\Gamma(i))=u_i(A_{ij}\cap \Gamma^{-1}(i)) + u_i(C_{ij}\cap \Gamma^{-1}(i)) + u_i(R_*\cap \Gamma^{-1}(i))\\
     u_i(\Gamma(j))=u_i(B_{ij}\cap \Gamma^{-1}(j)) + u_i(C_{ij}\cap \Gamma^{-1}(j)) + u_i(R_*\cap \Gamma^{-1}(j))
\end{eqnarray*}
If the scenario is a chore g-binary scenario, then
$u_i(A_{ij}\cap \Gamma^{-1}(i))=u_i(C_{ij}\cap \Gamma^{-1}(i))=u_i(C_{ij}\cap \Gamma^{-1}(j))=0$. While in a good g-binary scenario, $u_i(R_*\cap \Gamma^{-1}(i))=u_i(R_*\cap \Gamma^{-1}(j))=u_i(B_{ij}\cap \Gamma^{-1}(j))=0$.
Thus, the equations \eqref{eq-lema-1} and \eqref{eq-lema-2} are true.
\end{myproof}

}

\bigskip

\noindent{\bf Theorem 7.} {\em 
Under a g-binary scenario for goods or chores, the allocation $\Gamma$ given by Algorithm~\ref{algo:allocate} produces a maximal utilitarian welfare, is EF$1$ and PO. Its run time is $O(mn)$. 
}

\medskip

    \begin{myproof} \myred{We prove first the case of goods.}
We suppose that $\forall i\in \MA$, $u_i$ is a \myred{g-binary} utility function.
We consider the allocation $\Gamma$ of Algorithm~\ref{algo:allocate}. By Theorem \ref{main-thm-6}, $\Gamma$ maximizes utilitarian social welfare in $O(mn)$ operations, and by Theorem \ref{thm-main-1}, $\Gamma$ is PO.
We want to show that $\Gamma$ is EF1. Suppose that there exists   $i \in A$ such that
    \begin{equation}\label{eq4}
        u_{i}(\Gamma^{-1}(i)) < u_{i}(\Gamma^{-1}(j))
    \end{equation}
for some agent $j$. As $u_i$ is a \myred{g-binary} utility, by Lemma \ref{technicallemma}, it is enough to show that there exists $r \in \Gamma^{-1}(j) \cap C_{ij}$ such that
$u_{i}(\Gamma^{-1}(i)) \ge u_{i}(\Gamma^{-1}(j)) - u_{i}(r)$.
Clearly,  $\Gamma^{-1}(j) \cap C_{ij} \not= \emptyset$. Otherwise, by \eqref{eq-lema-2} and \eqref{eq4}, we have $u_{i}(\Gamma^{-1}(j)) < 0$, which is a contradiction because $u_i$ is non negative.
Let  $\Gamma^{-1}(j) \cap C_{ij} = \left\{r_{s_1},\dots,{r}_{s_{k}} \right\}$,
where
$r_{s_k}$ is the last resource assigned to $j$ and  it is  preferred by $i$ and $j$.
Since   $\Gamma(r_{s_k})=j$ and  $i\in P_{s_k}$,
by Remark~\ref{lemma-properties1-F}, part \ref{i},
\begin{equation*}\label{eq6}
     u_{i}(\Gamma^{-1}(i) \cap \left\{r_1,\dots,r_{s_{k}-1} \right\})
    \ge u_{j}(\Gamma^{-1}(j) \cap \left\{r_1,\dots,r_{s_{k}-1} \right\}).
\end{equation*}
By additivity of  $u_j$, 
    \begin{equation*}\label{eq7}
    u_{j}(\Gamma^{-1}(j) \cap \left\{r_1,\dots,r_{s_{k}-1} \right\} ) \ge u_{j}(\Gamma^{-1}(j) \cap \left\{r_1,\dots,r_{s_{k}-1} \right\} \cap C_{ij}).
\end{equation*}
Using the transitivity of  $\geq$, we have
    \begin{equation}\label{eq8}
    u_{i}(\Gamma^{-1}(i) \cap \left\{r_1,\dots,r_{s_{k}-1} \right\}) \ge u_{j}(\Gamma^{-1}(j) \cap \left\{r_1,\dots,r_{s_{k}-1} \right\} \cap C_{ij}).
    \end{equation}
As $u_i$ and $u_j$ are \myred{utility functions in a g-binary scenario}, then the utility that agents $ i $ and $ j $ assign to each resource in $ C_ {ij} $ is equal, so
    \begin{equation}\label{eq9}
        u_{j}(\Gamma^{-1}(j) \cap \left\{r_1,\dots,r_{s_{k}-1} \right\} \cap C_{ij}) = u_{i}(\Gamma^{-1}(j) \cap \left\{r_1,\dots,r_{s_{k}-1} \right\} \cap C_{ij}).
    \end{equation}
By equations \eqref{eq8}  and \eqref{eq9},
\begin{equation}\label{eq10}
   u_{i}(\Gamma^{-1}(i) \cap \left\{r_1,\dots,r_{s_{k}-1} \right\})\geq u_{i}(\Gamma^{-1}(j) \cap \left\{r_1,\dots,r_{s_{k}-1} \right\} \cap C_{ij}).
\end{equation}
Now,   by equation \eqref{eq-lema-2} and additivity of $u_i$,
    \begin{align*}
        u_{i}(\Gamma^{-1}(j)) &= u_{i}(\Gamma^{-1}(j) \cap C_{ij}) \\
                         &= u_{i}(\Gamma^{-1}(j) \cap \left\{r_1,\dots,r_{s_{k}-1} \right\} \cap C_{ij}) + u_{i}(r_{s_k})
    \end{align*}
    and, by equation \eqref{eq10},
\begin{equation*}
     u_{i}(\Gamma^{-1}(j))\leq u_{i}(\Gamma^{-1}(i) \cap \left\{r_1,\dots,r_{s_{k}-1} \right\})+ u_{i}(r_{s_k});
\end{equation*}
  so,
    $$u_{i}(\Gamma^{-1}(i) \cap \left\{r_1,\dots,r_{s_{k}-1} \right\}) \ge u_{i}(\Gamma^{-1}(j)) - u_{i}(r_{s_k}).$$
But, $u_{i}(\Gamma^{-1}(i)) \ge u_{i}(\Gamma^{-1}(i) \cap \left\{r_1,\dots,r_{s_{k}-1} \right\})$,
and using  transitivity,
\begin{equation}\label{eq-thm-main}
  u_{i}(\Gamma^{-1}(i)) \geq u_{i}(\Gamma^{-1}(j))- u_{i}(r_{s_k}).
\end{equation}
Therefore,  $\Gamma$ is  EF1.

\myred{Now we give the proof in the case of chores.
By Lemma \ref{technicallemma}, for all $i,j\in A$, \begin{equation}\label{eq1-thm}
 u_i(\Gamma^{-1}(i))=u_i(\Gamma^{-1}(i)\cap R_*)   
\end{equation}
and \begin{equation}\label{eq2-thm}
    u_i(\Gamma^{-1}(j))=u_{i}\left(\Gamma^{-1}(j) \cap B_{ij}\right) + u_{i}\left(\Gamma^{-1}(j) \cap R_{*}\right).
\end{equation}

If $R_*=\emptyset$, then   $u_i(\Gamma^{-1}(i))=0$ and $u_i(\Gamma^{-1}(j))=u_{i}\left(\Gamma^{-1}(j) \cap B_{ij}\right)\leq 0$. Thus, $u_i(\Gamma^{-1}(i))\geq u_i(\Gamma^{-1}(j))$. So, $\Gamma$ is EF. In consequence, $\Gamma $ is EF1.

 Suppose that  $R_*\neq \emptyset$ and that there exist $i,j\in A$ such that \begin{equation}\label{eq3-thm}
     u_i(\Gamma^{-1}(i))<u_i(\Gamma^{-1}(j)).
 \end{equation}
We want to show that there is $r\in \Gamma^{-1}(i)$ such that \begin{equation}\label{eq4-thm}
    u_i(\Gamma^{-1}(i)\setminus\{r\})\geq u_i(\Gamma^{-1}(j)
\end{equation}

 By equation  \eqref{eq3-thm},   $u_i(\Gamma^{-1}(i))<0$; and by equation  \eqref{eq1-thm}, $\Gamma^{-1}(i)\cap R_*\neq \emptyset $. 
 Using additivity   of $u_i$ and equations \eqref{eq1-thm} and \eqref{eq2-thm}, to show that \eqref{eq4-thm} is true, it is enough to find  $r\in\Gamma^{-1}(i)\cap R_*$ such that
 \begin{equation}\label{eq5-thm}
 u_i((\Gamma^{-1}(i)\cap R_*)\setminus\{r\})\geq u_{i}\left(\Gamma^{-1}(j)\right).
 \end{equation}
 Let  $r_{s}\in \Gamma^{-1}(i)\cap R_*$ such that  $r_{s}$ is the last chore allocated to $i$ and it is minimized  for all $l\in A$. 
Consider $\{r_1,\dots,r_{s-1}\}$\footnote{Note that if $r_{s}=r_1$, then $\{r_1,\dots, r_{s-1}\}=\emptyset$} the set of all chores allocated before $r_s$. As
$$\{r_1,\dots,r_{s-1}\}\cap \Gamma^{-1}(l)\cap R_*\subseteq \Gamma^{-1}(l)$$ then
\begin{equation}\label{eqxx-thm}
    u_i(R_*\cap \{r_1,\dots,r_{s-1}\}\cap \Gamma^{-1}(l) )\geq  u_i(\Gamma^{-1}(l)).
\end{equation}
On the other hand, as $r_s$ is the last chore assigned to $i$ and $r_s\in R_*$, then 
$$u_i(R_*\cap \Gamma^{-1}(i)\setminus\{r_s\})= u_i(R_*\cap \{r_1,\dots,r_{s-1}\}\cap \Gamma^{-1}(i))$$

and   by Lemma \ref{technicallemma}, 
$$ u_i(R_*\cap \{r_1,\dots,r_{s-1}\}\cap \Gamma^{-1}(i))= u_i( \{r_1,\dots,r_{s-1}\}\cap \Gamma^{-1}(i))$$
then, by transitivity, 
 \begin{equation}\label{eq10-thm}
     u_i(R_*\cap \Gamma^{-1}(i)\setminus\{r_s\})
   = u_i( \{r_1,\dots,r_{s-1}\}\cap \Gamma^{-1}(i)).
 \end{equation}
In order to complete the proof, we only have to show that 
\begin{equation}\label{eqxxx-thm}
 u_i( \{r_1,\dots,r_{s-1}\}\cap \Gamma^{-1}(i))\geq u_i(R_*\cap \{r_1,\dots,r_{s-1}\}\cap \Gamma^{-1}(j) )   
\end{equation}
because if in \eqref{eqxx-thm}, we take $l=j$ and   use  transitivity between \eqref{eq10-thm} , \eqref{eqxxx-thm} and \eqref{eqxx-thm},   the expression \eqref{eq5-thm} is true.

As $\Gamma(r_s)=i$, by the way the $\Gamma$ is defined, then 
$\forall l\in A$
\begin{equation*}
    u_i(\{r_1,\dots, r_{s-1}\}\cap \Gamma^{-1}(i))\geq u_l(\{r_1,\dots, r_{s-1}\}\cap \Gamma^{-1}(l)) ).
\end{equation*}
But, for all $l\in A$,
\begin{eqnarray*}\label{eq8-thm}
    u_l(\{r_1,\dots, r_{s-1}\}\cap \Gamma^{-1}(l) )&=u_l(R_*\cap\{r_1,\dots, r_{s-1}\}\cap \Gamma^{-1}(l))\\
    &=u_i(R_*\cap\{r_1,\dots, r_{s-1}\}\cap \Gamma^{-1}(l)) )
\end{eqnarray*}
and 
$$u_l(R_*\cap\{r_1,\dots, r_{s-1}\}\cap \Gamma^{-1}(l)) )=u_i(R_*\cap\{r_1,\dots, r_{s-1}\}\cap \Gamma^{-1}(l)) ).$$
Using transitivity,  we get \eqref{eqxxx-thm}. 
}
\end{myproof}

\medskip

\noindent{\bf Theorem 8.} {\em
Under a g-binary scenario for goods or chores, the allocation $\Gamma^*$ given by Algorithm~\ref{algo:allocate} modified as previously indicated maximizes the utilitarian social welfare, is  EFX and PO. Moreover, 
its run time is $O(m\log m + mn)$.
}

\medskip

\begin{myproof}
The argument of the proof is similar to that in the proof of the Theorem~\ref{thm:main}.  Let $i,j$ be in $\MA$ such that $$u_i({\Gamma^{*}}^{-1}(i))<u_i({{\Gamma^{*}}^{-1}}(j)).$$ 

\myred{First we consider the case of goods.} By Lemma~\ref{technicallemma},
$$u_{i}({\Gamma^{*}}^{-1}(j)) = u_{i}({\Gamma^{*}}^{-1}(j) \cap C_{ij} )=\{r_{s_1},\dots,r_{s_k}\}.$$
As $\MR$ is sorted in descending order, then for all $r\in {\Gamma^{*}}^{-1}(j)$ 
\myred{with $u_i(r)>0$},
$$ u_i(r_{s_k})\leq u_i(r).$$
By
equation~\eqref{eq-thm-main}
  $$u_{i}({\Gamma^{*}}^{-1}(i)) \geq u_{i}({\Gamma^{*}}^{-1}(j)) - u_{i}(r_{s_k}).$$
  So, for all $r\in {\Gamma^{*}}^{-1}(j)$ \myred{with $u_i(r)>0$},
  $$u_{i}({\Gamma^{*}}^{-1}(i)) \geq u_{i}({\Gamma^{*}}^{-1}(j)) - u_{i}(r).$$
  Thus, $\Gamma^{*}$ is $EFX$.

  \myred{Now,  we consider the case of chores. By Lemma~\ref{technicallemma},
\[
u_{i}({\Gamma^{*}}^{-1}(i)) = u_{i}({\Gamma^{*}}^{-1}(i) \cap \MR_*)=u_i({\Gamma^{*}}^{-1}(i)\cap \MR_*\cap  \{r_{1},\dots,r_{s}\})
\]
where $r_s$  is the last chore in $\MR_*$ allocated to $i$. Note that, $\forall r\in {\Gamma^{*}}^{-1}(i) \setminus \MR_*$, $u_i(r)=0$. As   $\MR$ is sorted in  increasing order
then for all $r\in{\Gamma^{*}}^{-1}(i)$ with $u_i(r)<0$,
$$ u_i(r_{s})\geq u_i(r).$$
Then \begin{equation}\label{eq-tech}
    u_i({\Gamma^{*}}^{-1}(i)) - u_{i}(r)\geq u_i({\Gamma^{*}}^{-1}(i)) - u_{i}(r_s).
\end{equation}

Since  \eqref{eq5-thm} is true for $r_s$,  using Transitivity, from   \eqref{eq-tech} and \eqref{eq5-thm}, we have that $\forall r\in {\Gamma^{*}}^{-1}(i)$ with $u_i(r)<0$, 
$$u_{i}({\Gamma^{*}}^{-1}(i)) - u_{i}(r) \geq u_{i}({\Gamma^{*}}^{-1}(j)).$$

  Thus, ${\Gamma^{*}}$ is $EFX$.
  }
\end{myproof}

\medskip 


\myred{
The following example shows the facts established in Remark~\ref{remark-new-behavior} about the behavior 
of $\Gamma$ and other 
allocations in a g-binary scenario.

\begin{ejem}\label{eje1}\label{exp:properties PO EF1}
Suppose that  $\MR=\{r_1,r_2,r_3,r_4,r_5\}$ and that  each resource  $r_k$ is valued as $p_k$ according to Table~\ref{tab:resource_values}.
\begin{table}[H]
     \caption{ $p_k$ values for each resource.}
     \label{tab:resource_values}
     \centering
     \begin{tabular}{cccccc}
  \hline
        & $r_1$&$r_2$&$r_3$&$r_4$&$r_5$  \\
         \hline
         $p_k$&500&200&50&100&250 \\
         \hline
         \end{tabular}
  \end{table}
  
Let  $\MA=\{1,2,3\}$ be the set of agents;  each agent  $i$ establishes its utility functions $u_i$ over each resource using  Table  \ref{tab:pref_of_agents_U3xR5}.

\begin{table}[H]
     \caption{Utility functions.}
     \label{tab:pref_of_agents_U3xR5}
     \centering
    \begin{tabular}{cccccc}
          \hline
          &$r_1$&$r_2$&$r_3$&$r_4$&$r_5$ \\ \hline
          $u_1$&500&200&50&0&0\\ 
          $u_2$&500&0&50&100&250\\
          $u_3$&500&200&0&100&0\\ \hline
    \end{tabular}
\end{table}

Note that all agents prefer the resource $r_1$ and their utility is 500. The resource $r_2$ is required by agents 1 and 3, and their utility is 200; the agents 1 and 2 give to resource $r_3$ the utility of 50.
Concerning
the resource $r_4$,  agents 2 and 3 give it the utility of 100; whereas resource $r_5$ is only required  by agent 2 with utility 250. If for  $i=1,2,3$, the function  $u_i$ is extended additively over each subset $S\subseteq\MR$, that is,  $\forall i\in \MA$,  $u_i(S)= \sum_{r\in S}u_i(\{r\})$, then each $u_i$ is a g-binary utility function, i.e., it is a g-binary scenario.

In this scenario, we will consider four allocations in order to illustrate their behavior with respect to   properties  EF1, PO, $MSW_u$ and $MSW_{Nash}$.
Let  $F$, $G$,  $\Gamma$  and $J$ be  the allocations defined by Table \ref{tab:def_of_allocations_R5xFJ}, which shows the agent number who receives each resource. 
\begin{table}[H]
     \caption{Allocations.}
     \label{tab:def_of_allocations_R5xFJ}
     \centering
    \begin{tabular}{cccccc}
  \hline
        & $r_1$&$r_2$&$r_3$&$r_4$&$r_5$  \\
         \hline
         $F$&1&3&3&3&2 \\
         $G$&1&3&2&3&2 \\
         $\Gamma$&1&3&2&2&2 \\
         $J$&1&1&2&3&2 \\ \hline
         \end{tabular}
\end{table}
In Table \ref{tab:tabla4}, we show the utility assigned by each agent to its  received bundle, and the social welfare in each allocation.
\begin{table}[H]
    \caption{Utility by received bundle and social welfare.}
     \label{tab:tabla4}
     \centering
    \begin{tabular}{cccccc}
    \hline
         & 1&2&3&$SW_u$&$SW_{Nash}$  \\
         \hline
         $u_i(F^{-1}(i))$&500 &250 &300&1050&37500000 \\
         $u_i(G^{-1}(i))$&500&300&300&1100&45000000 \\
          $u_{i}(\Gamma^{-1}(i))$&500&400&200&1100&40000000 \\
          $u_{i}(J^{-1}(i))$&700&300&100&1100&21000000 \\ \hline
         \end{tabular}
\end{table}

The properties of each allocation are described in Table \ref{tab:tabla5} and the detailed verification of these facts is given below. 
  
\begin{table}[H]
    \caption{Allocation properties under a \myred{g-binary} scenario.}
     \label{tab:tabla5}
     \centering
    \begin{tabular}{ccccc}
    \hline
        & PO&EF1&$MSW_u$&$MSW_{Nash}$  \\
         \hline
         $F$&\ding{53}&\ding{51}&\ding{53}&\ding{53}\\
         $G$&\ding{51}&\ding{51}&\ding{51}&\ding{51}\\
         $\Gamma$&\ding{51}&\ding{51}&\ding{51}&\ding{53}\\
         $J$&\ding{51}&\ding{53}&\ding{51}&\ding{53}\\
         \hline
         \end{tabular}
  \end{table}   

It is important to note that $F$ is an allocation EF1 and it is not in $MSW_u$. $G$ satisfies all the properties.
$\Gamma$ is EF1, PO, it  is in  $MSW_u$ but it is not in $MSW_{Nash}$.
Finally, $J$ is $PO$ and it is in $MSW_u$ but it is neither EF1 nor in $MSW_{Nash}$.

\end{ejem}

}

\begin{myproofof}{Properties in Example~\ref{exp:properties PO EF1}}\label{exam-cont-ejem1}
We give details of the properties fulfilled by the allocations  in Table \ref{tab:tabla5}.

From Theorem \ref{thm:utili-imp-MSWU} and Table \ref{tab:resource_values}, it is easy to see that the maximum utilitarian welfare is reached in 1100. Moreover;  $G$, $\Gamma$ and $J$  maximize $SW_u$. By Theorem~\ref{lema-SWU-impli-OP} and  Table \ref{tab:tabla4}, we have that $G$, $\Gamma$, and $J$ are PO.
On the other hand, from Table \ref{tab:tabla4}, we can observe that agents 1 and 3 in the allocations $F$ and $G$, have the same utility for the received bundle; but, in $G$, agent 2 improves its utility. Then, $F$ is Pareto dominated by $G$. Therefore, $F$ is not PO.

A search determined that the maximum Nash social welfare is reached at 45000000. Then $G$ is a maximum Nash social welfare and, by  Theorem \ref{thm-Caragiannis},  $G$ is EF1. Moreover, allocations  $F$ and $\Gamma$ are EF1. In fact,
 $\Gamma$ is EF$1$ by Theorem~\ref{main-thm-6}.
 
For $ F $,  agents $ 2 $ and $ 3 $ envy  agent $ 1 $, however, the envy disappears when eliminating $ r_1 $.
Finally, $J$ is not EF1, because  agent 3 envies  agent 1, $u_3(J^{-1}(3))<u_3(J^{-1}(1))$ and
$$u_3(J^{-1}(3)=100<200=u_3(J^{-1}(1)\setminus\{r_1\})<500=u_3(J^{-1}(1)\setminus\{r_2\}).$$
\end{myproofof}

\myred{The following example shows that $\Gamma^*$ is not EFX$_0$.}

\begin{ejem}\label{ej-EFX-not-EFXO}
Let's consider the following \myred{g-binary} scenario where $n=3$, $m=8$.
Table~\ref{CE-valuation} shows the utility functions that each agent gives to each resource.
\begin{table}[H]
     \caption{ Utility functions.}
     \label{CE-valuation}
     \centering
    \begin{tabular}{ccccccccc}
  \hline
        & $r_1$&$r_2$&$r_3$&$r_4$&$r_5$&$r_6$&$r_7$&$r_8$  \\
         \hline
         $u_1$&20&0&10&2&0&0&3&1 \\
         $u_2$&20&0&10&2&11&19&0&1 \\
         $u_3$&20&9&0&2&0&19&3&1 \\ \hline
         \end{tabular}
\end{table}

Now, using the modified Algorithm~\ref{algo:allocate}, we get the following allocation $\Gamma^*$.
\begin{table}[H]
     \caption{Allocation $\Gamma^*$}
     \label{tab:def_of_allocations_R8}
     \centering
    \begin{tabular}{ccccccccc}
  \hline
        & $r_1$&$r_2$&$r_3$&$r_4$&$r_5$&$r_6$&$r_7$&$r_8$  \\
         \hline
         $\Gamma^*$&1&3&1&3&2&2&3&3 \\ \hline
         \end{tabular}
\end{table}
Then, in Table \ref{CE-allocation_1} we show the utility assigned by each agent to its received bundle, and the social welfare in each allocation.
\begin{table}[H]
     \caption{Utility by received bundle and social welfare.}
     \label{CE-allocation_1}
     \centering
    \begin{tabular}{ccccccccc}
  \hline
        & $1$&$2$&$3$&$SW_{U}$&$SW_{Nash}$  \\
         \hline
         $u_{i}({\Gamma^*}^{-1}(i))$&30&30&15&75&13500 \\ \hline
         \end{tabular}
\end{table}
Notice that agent $3$ envies agent $2$ because 
\[u_{3}(\{r_2,r_5,r_6\})= 15\]
\[u_{3}(\{r_5,r_6\})=19\]
which means that 
\[u_{3}(\{r_2,r_5,r_6\}) < u_{3}(\{r_5,r_6\}).\]
Since $\Gamma^*$ was obtained using the modified Algorithm~\ref{algo:allocate}, we obtain that $\Gamma^*$ is $EFX$.
Now we check that $\Gamma^*$ does not satisfy $EFX_{0}$. Observe that
\[u_{3}(\{r_5,r_6\} \backslash \{r_5\})=19 \ \ \text{given that } \ \ u_3(\{r_5\})=0\]
so, it follows that
\begin{equation}\label{CE-last}
    u_{3}(\{r_2,r_5,r_6\})= 15 < 19=u_{3}(\{r_5,r_6\} \backslash \{r_5\})
\end{equation}

Thus,  property $EFX_{0}$ is not satisfied. 
\end{ejem}

\bibliography{citas}

\end{document}